\documentclass{article}
\usepackage[utf8]{inputenc}
\usepackage[english]{babel}
\usepackage[utf8]{inputenc}

\usepackage{lscape}
\usepackage{amssymb}
\usepackage{amsmath}
\usepackage{mathtools}
\usepackage{stmaryrd}
\usepackage[hypcap=false]{caption}
\usepackage{multirow}

\usepackage{amsthm}
\usepackage{makeidx}
\usepackage[color,all]{xy}
\usepackage{xcolor,colortbl}

\usepackage[hmargin={3cm,3cm},vmargin={2cm,2cm}]{geometry}

\usepackage[pdftex]{graphicx}
\usepackage{comment}
\usepackage{tikz}

\usepackage[colorlinks=true,linkcolor=black]{hyperref}

\usepackage{color}
\usepackage{tikz}
\usetikzlibrary{positioning}
\usetikzlibrary{matrix}

\hypersetup{urlcolor=blue, citecolor=red}

\newcounter{theorem}
\newtheorem{Theorem}{Theorem}[section]

\newtheorem{lemma}[theorem]{Lemma}

\theoremstyle{definition} 

\newtheorem{definition}[theorem]{Definition}

\begin{document}

\centerline{\Large \bf Geometric models for Lie--Hamilton systems on $\mathbb{R}^2$}

\vskip 0.20cm

\centerline{J. Lange and J. de Lucas}

\vskip 0.20cm

\centerline{Department of Mathematical Methods in Physics, University of Warsaw,}
\centerline{ul. Pasteura 5, 02-093, Warsaw, Poland}

\begin{abstract}
This paper provides a geometric description for  Lie--Hamilton systems on $\mathbb{R}^2$ with locally transitive Vessiot--Guldberg Lie algebras through two types of geometric models. The first one is the restriction  of a class of Lie--Hamilton systems on the dual of a Lie algebra to even-dimensional symplectic leaves relative to the Kirillov-Kostant-Souriau bracket. The second is a projection onto a quotient space of an automorphic Lie--Hamilton system relative to a naturally defined Poisson structure or, more generally, an automorphic Lie system with a compatible bivector field. These models give rise to a natural framework for the analysis of Lie--Hamilton systems on $\mathbb{R}^2$ while retrieving known results in a natural manner. Our methods may be extended to study Lie--Hamilton systems on higher-dimensional manifolds  and provide new approaches to Lie systems admitting compatible geometric structures.
\end{abstract}



	\section{Introduction}
	A {\it Lie system} is a first-order system of ordinary differential equations (ODEs)
	whose general solution can be written as a function, a so-called
	{\it superposition rule}, of a generic family of particular solutions and
	some constants to be related to initial conditions \cite{CGM00,CGM07,Dissertationes,LS,Ve93,Ve99,PW}. 
	Lie systems can be characterised by the  Lie--Scheffers theorem \cite{CGM00,CGM07,Dissertationes,LS,LS19,PW}, which states that a Lie system is equivalent to a $t$-dependent vector field taking values in a finite-dimensional Lie algebra
	of vector fields, called a {\it Vessiot--Guldberg (VG) Lie algebra} of the Lie system \cite{Dissertationes,LS19}.
	
	Every Lie system on $\mathbb{R}^2$ can be endowed with a VG Lie algebra that is, around a generic
	point, locally diffeomorphic to one of the twenty-eight possible classes of VG Lie algebras  on $\mathbb{R}^2$ described in \cite{GKO92,LS,LS19}. In particular, only twelve classes can be considered, again locally around a generic point, as VG Lie algebras of Hamiltonian vector fields relative to a symplectic form (see \cite{BHLS15,LS19} and Table \ref{table1}). Lie systems admitting a VG Lie algebra of Hamiltonian vector fields relative to a Poisson bivector are called {\it Lie--Hamilton systems} \cite{BHLS15,CLS13,LS19}. In the case of Lie--Hamilton systems on $\mathbb{R}^2$, the role played by the Poisson bivector is accomplished, at a generic point on the plane, by a symplectic form \cite{BBHLS15}.  
	
	Although Lie--Hamilton systems on $\mathbb{R}^2$ are the exception rather than the rule among general differential equations (cf.\! \cite{BHLS15,LS19}), they admit a plethora of geometric properties and relevant applications \cite{BCFHL17,BCFHL18,BCHLS13,BHLS15,LS19}, which motivates their analysis. For instance, Smorodinsky--Winternitz oscillators \cite{FMSUW65} or certain diffusion equations can be analysed via Lie--Hamilton systems on $\mathbb{R}^2$ (see \cite{BCHLS13,BHLS15,LS19} and references therein).
	
	    Among Lie systems, a relevant role is played by  {\it automorphic Lie systems}, i.e. a class of Lie systems described by a $t$-dependent vector field on a Lie group $G$ taking values in its Lie algebra of right-invariant vector fields (or in its Lie algebra of left-invariant vector fields either) \cite{CGM00,Dissertationes,LS19}. All Lie systems can be described via automorphic Lie systems (see \cite{CGM00,Dissertationes,LS19,SW84,SW84II}). Nevertheless, there exists no prototypical geometric model to describe Lie--Hamilton systems. Moreover, many works on Lie--Hamilton systems try to derive or to explain the existence of a Poisson bivector or symplectic form turning the elements of a VG Lie algebra into Hamiltonian vector fields \cite{BHLS15,LL18}. This has been done by solving systems of partial differential equations (PDEs) \cite{BBHLS15} or using other algebraic and geometric techniques \cite{BHLS15,BCFHL18,LL18}. 
	
	The aim of this work is to provide prototypical geometric models for Lie--Hamilton systems on $\mathbb{R}^2$ admitting {\it locally transitive}  VG Lie algebras, i.e.  Lie algebras of vector fields whose elements span, at every point, the tangent space to the manifold where they are defined on.  More specifically, we focus on the classes P$_1$, P$_2$, P$_3$, P$_5$, I$_4$, I$_5$, I$_8$, I$_{14A}$, I$_{14B}$, and I$_{16}$ given in Table \ref{table1}. We use two types of geometric models for studying Lie--Hamilton systems on $\mathbb{R}^2$, and we propose methods to derive geometrically their associated symplectic structures.
	
	Our first model describes Lie--Hamilton systems on $\mathbb{R}^2$ as the restriction of a Lie--Hamilton system on the dual, $\mathfrak{g}^*$, of a Lie algebra $\mathfrak{g}$  to an even-dimensional symplectic leaf on $\mathfrak{g}^*$ relative to the Poisson structure given by the  {\it Kirillov-Kostant-Souriau (KKS) bracket} \cite{Va94}. Our methods significantly extend and clarify the procedure briefly sketched in \cite[Theorem 2]{BCFHL18}, which was focused on studying only a particular simple example and it did not consider several difficulties for the application of the technique. Moreover, we now obtain as a byproduct that certain  Lie--Hamilton systems on $\mathbb{R}^2$ are endowed with a pseudo--Riemannian metric that is invariant relative to the Lie derivative with respect to the vector fields of the VG Lie algebra of the Lie--Hamilton system. This explains previous results relative to the existence of Lie systems on $\mathbb{R}^2$ admitting a VG Lie algebra of Killing vector fields with respect to a pseudo-Riemannian metric  \cite{LL18}. 
	
	Our second model is a projection onto a quotient space of an automorphic Lie system  whose VG Lie algebra leaves invariant a Poisson bracket, or a  general bivector field that it is also projectable. This model extends the methods in \cite{GLMV19}, where invariant tensor fields relative to VG Lie algebras for Lie systems were constructed via $\mathfrak{g}$-invariant elements of tensor algebras on the elements of a Lie algebra. 

	It is remarkable that all presented geometric models can be immediately extended to describe Lie--Hamilton systems on higher-dimensional manifolds. Moreover, our results also provide additional information on the existence of symplectic leaves for KKS brackets related to general Lie algebras, e.g. the existence of two-dimensional symplectic leaves relative to the KKS bracket on $\mathfrak{sl}^*_3$. Finally, our techniques give rise to the conjecture that the compatible structures for Lie systems can be obtained very generally via the projections onto quotient spaces of automorphic Lie systems admitting VG Lie algebras of right-invariant vector fields and associated left-invariant tensor fields. 
	
	\begin{table}[t] {\footnotesize
 \noindent
\caption{{\small Classes of VG Lie algebras of Hamiltonian vector fields on $\mathbb{R}^2$ (see \cite{BHLS15,GKO92,HA75}). In particular, every Lie algebra consists of Hamiltonian vector fields on the submanifold of $\mathbb{R}^2$ where the given symplectic form is well defined. The functions $\xi_1(x),\ldots,\xi_r(x)$ and $1$ are any set of linearly independent functions, while  $\eta_1(x),\ldots,\eta_r(x)$ form a basis of fundamental solutions for an $r$-order homogeneous differential equation with constant coefficients \cite[pp.~470--471]{HA75}. Finally, $\mathfrak{g}_1\ltimes \mathfrak{g}_2$ stands for the semi-direct product  of the Lie algebras $\mathfrak{g}_1$ and $\mathfrak{g}_2$, where $\mathfrak{g}_2$ becomes an ideal of $\mathfrak{g}_1\ltimes \mathfrak{g}_2$. We assume that $i=1,\ldots,r$ and $r\geq 1$. Locally transitive VG Lie algebras are written in bold.}}
\label{table1}
\medskip
\noindent\hfill
 \begin{tabular}{ p{1.4cm} p{3cm}    p{7.7cm} p{1.7cm} l}
\hline
&  &\\[-1.9ex]
Primitive&Lie algebra structure & Basis of vector fields $X_i$ & $\omega$ \\[+1.0ex]
\hline
 &  &\\[-1.9ex]
{\bf P$_1$}&$\mathfrak{iso}_2\simeq \mathbb{R}\ltimes \mathbb{R}^2$ & $  { {\partial_x} ,    {\partial_y} ,      y\partial_x - x\partial_y},$ &$dx\wedge dy$\\[+1.0ex]
{\bf P$_2$}&$\mathfrak{sl}_2$ (type I) & $ {\partial_x},   {x\partial_x  +  y\partial_y} ,   (x^2  -  y^2)\partial_x  +  2xy\partial_y$&$\frac{dx\wedge dy}{y^2}$\\[+1.0ex]
{\bf P$_3$}&$\mathfrak{so}_3$ &${     { y\partial_x  \!-\!  x\partial _y},     { (1  +  x^2 \! - \! y^2)\partial_x \! +\!  2xy\partial_y} ,   2xy\partial_x  \!+ \! (1  +  y^2\!  - \! x^2)\partial_y}$&$\frac{dx\wedge dy}{1+x^2+y^2}$\\[+1.0ex]
{\bf P$_5$}&$\mathfrak{sl}_2\ltimes\mathbb{R}^2$ &${  {\partial_x},   {\partial_y},  x\partial_x - y\partial_y,  y\partial_x,  x\partial_y}$&$dx\wedge dy$\\[+1.0ex]
\hline
&  \\[-1.5ex]
Imprimitive & Lie algebra structure\!\! & Basis of vector fields $X_i$ & $\omega$ \\[+1.0ex]
\hline
&  \\[-1.5ex]
I$_1$&$\mathbb{R}$ &$  {\partial_x} $&$f(y)dx\wedge dy$ \\[+1.0ex]
{\bf I$_4$}&$\mathfrak{sl}_2$ (type II) & ${ {\partial_x  +  \partial_y},    {x\partial _x + y\partial_y},   x^2\partial_x  +  y^2\partial_y}$ &$\frac{dx\wedge dy}{(x-y)^2}$\\[+1.0ex]
{\bf I$_5$}&$\mathfrak{sl}_2$ (type III) &${ {\partial_x},    {2x\partial_x + y\partial_y},   x ^2\partial_x  +  xy\partial_y}$&$\frac{dx\wedge dy}{y^3}$\\[+1.0ex]
{\bf I$_8$}&$\mathfrak{iso}_{1,1}\simeq \mathbb{R}\ltimes\mathbb{R}^2$ &${  {\partial_x},    {\partial_y},   x\partial_x  - y\partial_y}\quad  $&$dx\wedge dy$\\[+1.0ex]
I$_{12}$&$\mathbb{R}^{r + 1}$ &$ {\partial_y} ,   \xi_1(x)\partial_y, \ldots , \xi_r(x)\partial_y$&$f(x)dx\wedge dy$\\[+1.0ex]
{\bf I$_{14A}$}&$\mathbb{R}\ltimes \mathbb{R}^{r}$ & ${ {\partial_x},   {\eta_1(x)\partial_y} ,  {\eta_2(x)\partial_y},\ldots ,\eta_r(x)\partial_y},\quad 1\neq \eta_i(x)$&$dx\wedge dy$\\[+1.0ex]
{\bf I$_{14B}$}&$\mathbb{R}\ltimes \mathbb{R}^{r}$ & ${ {\partial_x},   {\eta_1(x)\partial_y} ,  {\eta_2(x)\partial_y},\ldots ,\eta_r(x)\partial_y},\quad 1=\eta_1(x)$&$dx\wedge dy$\\[+1.0ex]
{\bf I$_{16}$}&$C_{-1}^r\simeq \mathfrak{h}_2\ltimes\mathbb{R}^{r + 1}$ & ${  {\partial_x},    {\partial_y} ,   x\partial_x  -y\partial_y,   x\partial_y, \ldots, x^r\partial_y},\quad    $&$dx\wedge dy$\\[+1.0ex]
\hline
 \end{tabular}
\hfill}
\end{table}

	\section{Fundamentals}
	
	Let us provide a brief account of the theory of Lie--Hamilton systems needed to understand our work and to make our exposition almost self-contained. To highlight our main
ideas and to avoid minor technical details, we assume, if not otherwise stated, that mathematical objects are smooth and well defined globally (see \cite{AM87,CGM00,Dissertationes,LS19,Va94} for further details). Hereafter, $N$ stands for a connected  $n$-dimensional manifold, and $\mathfrak{g}$ denotes an abstract finite-dimensional Lie algebra.
	
	We call a $t$-dependent vector field on $N$ a map $X:(t,x)\in \mathbb{R}\times N\mapsto X(t,x)\in TN$ satisfying that $\tau_N\circ X=\pi_2$, where $\tau_N:TN\rightarrow N$ is the canonical tangent bundle projection and we define $\pi_2:(t,x)\in \mathbb{R}\times N\mapsto x\in N$. An {\it integral curve} of $X$ is a solution $\gamma:\mathbb{R}\rightarrow N$ to
	\begin{equation}\label{Asso}
	\frac{d\gamma}{dt}(t)=X(t,\gamma(t)),\qquad \forall t\in \mathbb{R}.
	\end{equation}
	Then,  $\widetilde{\gamma}:t\in \mathbb{R}\mapsto (t,\gamma(t))\in \mathbb{R}\times N$ becomes an integral curve of the {\it autonomisation} (also called {\it suspension}) of $X$, i.e. the vector field on $\mathbb{R}\times N$ of the form $\widetilde{X}:=\partial_t+X$ \cite{AM87,Dissertationes}. Conversely, a section $\widetilde{\gamma}:\mathbb{R}\rightarrow \mathbb{R}\times N$ of the bundle $\pi_1:(t,x)\in \mathbb{R}\times N\mapsto t\in \mathbb{R}$ that is additionally an integral curve of $\widetilde{X}$ leads to a solution $\pi_2\circ \widetilde{\gamma}$ to (\ref{Asso}). This one-to-one correspondence enables us to identify  system  (\ref{Asso}) with its associated $t$-dependent vector field $X$. This also allows us to shorten the terminology of the paper.
	
	Each $t$-dependent vector field $X$ on $N$ amounts to a $t$-parametric family of standard vector fields on $N$ of the form $\{X_t:x\in N\mapsto X(t,x)\in TN\}_{t\in \mathbb{R}}$. We call {\it smallest Lie algebra} of $X$ (also called {\it minimal Lie algebra} or {\it irreducible Lie algebra} in the literature \cite{LS19}) the smallest Lie algebra (in the sense of inclusion) of vector fields, $V^X$, containing $\{X_t\}_{t\in \mathbb{R}}$. Let us denote by $\mathcal{D}^{V}$ the generalised distribution on $N$ spanned by the elements of a Lie algebra of vector fields $V$. Then, $\mathcal{D}^{V}$ is regular in the connected components of an open and dense subset of $N$ (see \cite{Pa57,Va94} for details). Special attention will be hereafter paid to finite-dimensional Lie algebras of vector fields, the so-called {\it Vessiot--Guldberg (VG) Lie algebras}.

	A {\it superposition rule} \cite{CGM07,Dissertationes,LS19,PW} for a system $X$ on  $N$ is a map $\Psi:N^m\times N\rightarrow N$ satisfying that the general solution, $x(t)$, to $X$ can be expressed as
	$$
	x(t)=\Psi(x_{(1)}(t),\ldots, x_{(m)}(t),k),
	$$
	for a fixed family of particular solutions $x_{(1)}(t),\ldots,x_{(m)}(t)$ to $X$ and a parameter $k\in N$ to be related to the initial conditions of $x(t)$. We call {\it Lie system} a system of ODEs admitting a superposition rule. Although the term superposition rule has been used in the literature with different meanings and there exist different approaches to each notion, our definition in this paper is the predominant in the literature on nonlinear differential equations (cf. \cite{CGM00,Dissertationes,LS19,Ma01,Um89} and references therein).
	
	\begin{Theorem}{ (The Lie--Scheffers theorem \cite{CGM00,CGM07,Dissertationes, LS,PW})} A system $X$ on $N$ admits a
		superposition rule if and only if 
		$
		X=\sum_{\alpha=1}^rb_\alpha(t)X_\alpha
		$
		for a set $X_1,\ldots,X_r$ of vector fields on $N$ generating an $r$-dimensional Lie algebra of vector fields, called a {\it VG Lie algebra} of $X$, and
		a family $b_1(t),\ldots,b_r(t)$  of $t$-dependent functions.
	\end{Theorem}

		A prototypical type of Lie system (see \cite{CGM00,Dissertationes,LS19}) is given by the system of differential equations on a Lie group $G$ associated with a $t$-dependent vector field 
	\begin{equation}\label{Aut}
	X^G(t,g):=\sum_{\alpha=1}^rb_\alpha(t)X^R_\alpha(g),\qquad \forall g\in G,
	\end{equation}
	where $X^R_1,\ldots,X_r^R$ stand for a basis of right-invariant vector fields
	on $G$, and $b_1(t),\ldots,b_r(t)$ are arbitrary $t$-dependent functions. Indeed, $X^R_1,\ldots,X^R_r$ span an $r$-dimensional VG Lie algebra on $G$ \cite{AM87}. The Lie--Scheffers theorem yields then that (\ref{Aut}) admits a superposition rule, which turns $X^G$ into a Lie system. In fact, the right-invariance of (\ref{Aut}) relative to the action of $G$ on itself from the right ensures that if $g_0(t)$ is any particular solution to (\ref{Aut}) and $h\in G$, then $g_0(t)h$ is another particular solution to (\ref{Aut}). Since initial conditions to (\ref{Aut}) are in one-to-one correspondence with its particular solutions, the general solution to (\ref{Aut}), let us say $g(t)$, can be written as
	$$
	g(t)=g_0(t)h,
	$$
	where $h$ is an arbitrary element of $G$. 
	Then, $X^R$ possesses a superposition rule $\Psi:(g,h)\in G\times G\mapsto gh\in G$. Lie systems taking the form (\ref{Aut}) are called {\it automorphic Lie
		systems} \cite{Dissertationes,GLMV19}. Every Lie system can be solved through one particular solution to a related automorphic Lie system \cite{CGM00,Dissertationes,LS19}. Automorphic Lie systems play a relevant role in determining superposition rules for special classes of Lie systems 
		\cite{SW84,SW84II}.
		
		Let us turn to studying Lie systems with compatible Poisson structures (see \cite{LS19} for details). 
		We recall that a {\it Poisson structure} on $N$ is a map $\{\cdot,\cdot\}:C^\infty(N)\times C^\infty(N)\rightarrow C^\infty(N)$ that is  antisymmetric, bilinear, and it satisfies the Leibniz property and the Jacobi identity. The pair $(N,\{\cdot,\cdot\})$ is called a {\it Poisson manifold}. We just say that $N$ is a Poisson manifold if $\{\cdot,\cdot\}$ is understood by context. Since $\{\cdot,\cdot\}$ is a derivation on each entry, it amounts to a bivector field $\Lambda$ on $N$ satisfying $[\Lambda,\Lambda]_{SN}=0$, where $[\cdot,\cdot]_{SN}$ stands for the {\it Schouten-Nijenhuis bracket}. Conversely, a bivector field $\Lambda$ on $N$ satisfying $[\Lambda,\Lambda]_{SN}=0$ leads to a Poisson structure $\{f,g\}:=\Lambda(df,dg)$ for every $f,g\in C^\infty(N)$. Hence, $\{\cdot,\cdot\}$ and its associated $\Lambda$ can be considered as equivalent.
		
		A vector field $X$ on a Poisson manifold $N$ is {\it Hamiltonian} if there exists a certain $\widetilde{f}\in C^\infty(N)$ such that $Xf=\{\widetilde{f},f\}$ for every $f\in C^\infty(N)$. A {\it Casimir function} of $\{\cdot,\cdot\}$ is a function $h\in C^\infty(N)$ satisfying that $\{h,f\}=0$ for every $f\in C^\infty(N)$.
		
		The use of Poisson structures in the study of Lie systems appeared very succinctly in \cite{CGM00}, but its usefulness, for instance so as to obtain superposition rules and constants of motion for Lie--Hamilton systems, was shown posteriorly in \cite{CLS13}. In particular, the use of Poisson geometry in the study of Lie systems is based on the definition below \cite{CLS13,LS19}.
		
	\begin{definition} A system $X$ on $N$ is called a {\it Lie--Hamilton system} if
$V^X$ is a VG Lie algebra of Hamiltonian vector fields relative to a Poisson bivector on $N$.
\end{definition}

	The work \cite{BBHLS15} showed that every VG Lie algebra $V$ of Hamiltonian vector fields on a two-dimensional Poisson manifold $N$ is, around each point $p\in N$ where $\mathcal{D}^V_p=T_pN$, locally diffeomorphic to a VG Lie algebra detailed in just one of the twelve classes of Table \ref{table1}. Note that I$_{12}$, I$_{14A}$, I$_{14B}$, and I$_{16}$
	are additionally subdivided into subclasses depending on the index $r$, which fixes the dimension of the VG Lie algebras of each subclass. 
	
	VG Lie algebras may be isomorphic as Lie algebras without being locally diffeomorphic, e.g. this concerns the VG Lie algebras of the classes P$_2$, I$_4$, and I$_5$  (see Table \ref{table1}). Meanwhile, VG Lie algebras belonging to one of the classes I$_{12}$, I$_{14A}$, I$_{14B}$, and I$_{16}$,  may not be isomorphic between themselves as Lie algebras. For instance, $\langle \partial_x,e^x\partial_y,e^{2x}\partial_ y\rangle$ and $\langle \partial_x,e^x\partial_y,xe^x\partial_y\rangle$ belong to I$_{14A}$ (for $r=2$), but they are not isomorphic as Lie algebras because the first one is isomorphic to $\mathbb{R}\ltimes_1\mathbb{R}^2$ while the second is isomorphic to $\mathbb{R}\ltimes_2\mathbb{R}^r$ for a different, not equivalent to $\ltimes_1$, semi-direct product $\ltimes_2$.

If a Poisson bivector $\Lambda$ satisfies that the map $\widehat{\Lambda}:\alpha\in T^*N\mapsto \Lambda(\alpha,\cdot)\in TN$ is an isomorphism, then $\Lambda(\widehat{\Lambda}^{-1}\cdot,\widehat{\Lambda}^{-1}\cdot)$ becomes a symplectic form on $N$.
Every Poisson manifold $(N,\Lambda)$ is such that $N$ can be decomposed into the sum of non-intersecting submanifolds (generally of  different dimension), a so-called {\it stratification} \cite{AM87}, in such a way that $\Lambda$ is tangent to each submanifold and its restriction induces a symplectic structure on it. Moreover, the tangent space to each submanifold of the stratification is spanned by the restriction to it of the Hamiltonian vector fields relative to $\Lambda$ (see \cite{Va94} for details). 

Let us now define a new structure for the study of Lie systems. Let $G$ be a Lie  group and let $\mathcal{T}^L(G)$ be the space of left-invariant contravariant tensor fields on $G$. It is clear that if $\{X^L_1,\ldots, X^L_r\}$ is a basis of left-invariant vector fields on $G$, then every element of $\mathcal{T}^L(G)$ can be considered as a linear combination of tensor products of $X^L_1,\ldots, X^L_r$. Moreover, if $V^L$ stands for the space of left-invariant vector fields on $G$, then one can define a Lie algebra morphism $\mu:X\in V^L\mapsto \mu_X:=\mathcal{L}_X\in {\rm End}(\mathcal{T}^L(G))$, where $\mathcal{L}_X$ denotes the Lie derivative relative to $X$. 

Let us briefly analyse the space of invariants, $$
\mathcal{T}^L_{\rm inv}(G):=\{T\in \mathcal{T}^L(G):\mu_X(T)=0,\forall X\in V^L\},
$$
relative to  $\mu$. Let us show how Casimir elements of a Lie algebra $\mathfrak{g}$ and the  $\mathfrak{g}$-invariant elements of the Grassmann algebra, $\bigwedge \mathfrak{g}$, of $\mathfrak{g}$ can be considered as elements in $\mathcal{T}^L(G)$ (see \cite{Va84} for details). A Lie algebra $\mathfrak{g}$ admits a so-called {\it universal enveloping algebra}, $U(\mathfrak{g})$, which is linearly isomorphic to the space of symmetric tensor products of elements of $\mathfrak{g}$. The adjoint representation ${\rm ad}:v\in \mathfrak{g}\mapsto {\rm ad}_v\in {\rm End}(\mathfrak{g})$,  can be extended to a Lie algebra morphism ${\rm ad}:\mathfrak{g}\rightarrow {\rm End}(U(\mathfrak{g}))$ by considering that the extension of each ${\rm ad}_v:\mathfrak{g}\rightarrow \mathfrak{g}$  with $v\in \mathfrak{g}$ to $U(\mathfrak{g})$, namely ${\rm ad}_v:U(\mathfrak{g})\rightarrow U(\mathfrak{g})$, is a derivation relative to the tensor product. Then, a {\it Casimir element} of $\mathfrak{g}$ is a $C\in U(\mathfrak{g})$ satisfying that ${\rm ad}_vC=0$ for every $v\in\mathfrak{g}$. Similarly, $\bigwedge \mathfrak{g}$ is linearly isomorphic to the space of skew-symmetric tensor products of elements of $\mathfrak{g}$. Likewise, ${\rm ad}:\mathfrak{g}\rightarrow {\rm End}(\mathfrak{g})$ can be extended to a new Lie algebra morphism ${\rm ad}:\mathfrak{g}\rightarrow {\rm End}(\bigwedge \mathfrak{g})$ by extending each ${\rm ad}_v:\mathfrak{g}\rightarrow \mathfrak{g}$, with $v\in \mathfrak{g}$, to a derivation in $\Lambda\mathfrak{g}$ relative to the tensor product. A {\it $\mathfrak{g}$-invariant element} is an $L\in \bigwedge \mathfrak{g}$ such that ${\rm ad}_vL=0$ for every $v\in\mathfrak{g}$.

Since every abstract Lie algebra $\mathfrak{g}$ is isomorphic to the Lie algebra of left-invariant vector fields of a connected and simply connected Lie group $G$, Casimir elements and $\mathfrak{g}$-invariant elements amount to elements of $\mathcal{T}^L_{\rm inv}(G)$. A similar construction to the above one can be accomplished by considering right-invariant objects on $G$. We denote then by $\mathcal{T}^R_{\rm inv}(G)$ the space of right-invariant tensor fields on $G$. 

Our previous construction is similar to the tensor algebra structure given in \cite{GLMV19}. In this paper, we will use such an approach to study Lie--Hamilton systems on the plane or higher-dimensional manifolds.

	There exist other kinds of Lie systems admitting compatible geometric structures, e.g. a Lie system $X$ on $N$ is called a {\it pseudo-Riemannian Lie system}  when it possesses a VG Lie algebra of Killing vectors relative to a pseudo-Riemannian metric on $N$ (see \cite{LS19} and references therein).
	
	Finally, we detail a lemma that, although being rather immediate, it seems to be absent in the literature. It will be useful to study the description of Lie--Hamilton systems through automorphic Lie systems with a compatible geometric structure.
	
	\begin{lemma}
	\label{lemma3}
	A $k$-vector field $\mathfrak{V}$ on a Lie group $G$, i.e. an antisymmetric $k$-contravariant tensor field on $G$, can be projected onto the quotient space $G/H$ of left cosets of a connected Lie subgroup $H\subset G$ if and only if $\pi_*\mathcal{L}_{X_H^L}\mathfrak{V}=0$, where $\pi:G\rightarrow G/H$ is the canonical projection onto the quotient space, for every left-invariant vector field $X_H^L$ tangent to $H$. 
	\end{lemma}
	\begin{proof} The projection of $\mathfrak{V}$ to $G/H$ exists if and only if $\pi_{*g}\mathfrak{V}_g$ is the same for all those $g\in G$ projecting onto the same coset in $G/H$. This implies that every integral curve $\phi(t)$ with $\phi(0)=g$ of a vector field $X_H^L$ tangent to $H$ satisfies that
	$$
	\pi_{*\phi(t)}\mathfrak{V}_{\phi(t)}=\pi_{*g}\mathfrak{V}_g,\quad \forall t\in \mathbb{R}\quad\Leftrightarrow\quad \frac{d}{dt}\pi_{*\phi(t)}\mathfrak{V}_{\phi(t)}=0,\quad \forall t\in \mathbb{R}.
	$$
	Hence, for every $f^1,\ldots,f^k\in C^\infty(G/H)$ and $t\in \mathbb{R}$, we have
	$$
\frac{d}{dt}\pi_{*\phi(t)}\mathfrak{V}_{\phi(t)}((df^1)_{\pi(g)},\ldots,(df^k)_{\pi(g)})=0\Leftrightarrow \frac{d}{dt}\mathfrak{V}_{\phi(t)}((d\pi^*f^1)_{\phi(t)},\ldots,(d\pi^*f^k)_{\phi(t)})=0.
	$$
	Then, $0=\mathcal{L}_{X_H^L}\mathfrak{V} (d\pi^*f^1,\ldots,d\pi^*f^k)=(\pi_*\mathcal{L}_{X_H^L}\mathfrak{V})(df^1,\ldots,df^k)$. Hence, $\pi_* \mathcal{L}_{X_H^L}\mathfrak{V}=0$ for every $X^L_H$ tangent to $H$. This shows that the latter condition is necessary. Since the vector fields, $X_H^L$, tangent to $H$ span the whole tangent space to $H$ and $H$ is connected, the previous reasoning can be reversed to get that the given condition is sufficient.
	\end{proof}
	
	\section{Lie--Hamilton systems on $\mathbb{R}^2$ related to simple VG Lie algebras}\label{LHSg*}
	Let us show how a Lie--Hamilton system on $\mathbb{R}^2$ admitting a simple VG Lie algebra  can be considered as the restriction  of a certain type of Lie--Hamilton system on the dual to a Lie algebra to even-dimensional symplectic leaves of the KKS bracket on such a dual \cite{Va94}. Our approach extends the results and applications given in \cite[Theorem 2]{BCFHL18}.
	
	Let $\mathfrak{g}$ be a Lie algebra with a Lie bracket $[\cdot,\cdot]:\mathfrak{g}\times \mathfrak{g}\rightarrow \mathfrak{g}$. If $f\in C^\infty(\mathfrak{g}^*)$, where $\mathfrak{g}^*$ is the dual space to $\mathfrak{g}$, the canonical isomorphisms $\mathfrak{g}\simeq\mathfrak{g}^{**}$ and $T_\theta \mathfrak{g}^*\simeq \mathfrak{g}^*$, for any $\theta\in \mathfrak{g}^*$, allow us to consider $df_\theta:T_\theta \mathfrak{g}^*\rightarrow \mathbb{R}$ as a vector in $\mathfrak{g}\simeq \mathfrak{g}^{**}\simeq T^*_\theta \mathfrak{g}^*$. This leads to a Poisson bracket on $\mathfrak{g}^*$, called the {\it KKS bracket} \cite{Va94}, given by
	\begin{equation}\label{KKS}
	\{f,g\}(\theta):=\langle\theta, [df_\theta,dg_\theta]\rangle,\qquad \forall \theta\in \mathfrak{g}^*,\forall f,g\in C^\infty(\mathfrak{g}^*),
	\end{equation}
	where $\langle \theta,v\rangle$ stands for the value of the linear form $\theta\in \mathfrak{g}^*$ at $v\in \mathfrak{g}$. 
	
	Let us consider a basis $\{e_1,\ldots,e_r\}$ of $\mathfrak{g}$. The isomorphism $\mathfrak{g}\simeq\mathfrak{g}^{**}$ enables us to understand the elements of this basis as a coordinate system on $\mathfrak{g}^*$. Let us now define a differential equation on $\mathfrak{g}^*$ given by
	\begin{equation}\label{g*}
	\frac{d\theta_\beta}{dt}= \{h(t),e_\beta\},\qquad h(t):=\sum_{\alpha=1}^rb_\alpha (t)e_\alpha,\qquad \beta=1,\ldots,r,
	\end{equation}
	where $\theta_\beta=\langle \theta,e_\beta\rangle$ for $\beta=1,\ldots,r$, and $b_1(t),\ldots,b_r(t)$ are
	 arbitrary $t$-dependent functions. The vector fields  $X_\alpha f:=\{e_\alpha,f\}$ for every $f\in C^\infty(\mathfrak{g}^*)$ and $\alpha=1,\ldots,r$, span a finite-dimensional Lie algebra, $V_{\mathfrak{g}^*}$, since
	\begin{equation*}
	[X_\alpha,X_\beta]f=\{e_\alpha,\{e_\beta,f\}\}-\{e_\beta,\{e_\alpha,f\}\}=\{\{e_\alpha,e_\beta\},f\}=\sum_{\gamma=1}^rc_{\alpha\beta}\,^\gamma X_\gamma f,\quad \forall f\in C^\infty(\mathfrak{g}^*),
	\end{equation*}
	where $c_{\alpha\beta}\,^\gamma$, with $\alpha,\beta,\gamma=1,\ldots,r$, are the structure constants of $\mathfrak{g}$ in the chosen basis. Hence, $V_{\mathfrak{g}^*}=\langle X_1,\ldots,X_r\rangle$ is a Lie algebra and $\rho:v\in \mathfrak{g}\mapsto X_v:=\{v,\cdot\}\in V_{\mathfrak{g}^*}$ is a Lie algebra morphism. Moreover, $\ker \rho=\{v\in \mathfrak{g}:\{v,e_\alpha\}=[v,e_\alpha]=0,\alpha=1,\ldots,r\}$, namely $\ker\rho=Z(\mathfrak{g})$, where $Z(\mathfrak{g})$ denotes the centre of $\mathfrak{g}$.
	
	In view of the above comments, (\ref{g*}) is the system of differential equations associated with
	$$
	X_{\mathfrak{g}^*}:=\sum_{\alpha=1}^rb_\alpha(t)X_\alpha,
	$$
	which is a Lie system admitting a VG Lie algebra $V_{\mathfrak{g}^*}\simeq \mathfrak{g}/Z(\mathfrak{g})$. Moreover, since $Z(\mathfrak{g})=0$ due to the assumption on the simplicity of $\mathfrak{g}$, one obtains that $V_{\mathfrak{g}^*}\simeq \mathfrak{g}$. By construction, $X_1,\ldots,X_r$ are Hamiltonian relative to the KKS bracket (\ref{KKS}) on $\mathfrak{g}^*$. Hence, (\ref{g*}) becomes a Lie--Hamilton system.
	
	Recall that a Poisson structure on $N$ gives rise to a stratification of $N$ into symplectic submanifolds \cite{Va94}. In particular, $\mathfrak{g}^*$ admits a stratification, $\mathfrak{F}$, induced by the KKS bracket. The tangent space to each leaf of the stratification is spanned by the Hamiltonian vector fields of the KKS bracket \cite{Va94}. Consequently, $X_1,\ldots,X_r$ are tangent to each leaf $\mathfrak{F}_k$ of $\mathfrak{F}$, and system (\ref{g*}) can be restricted to each $\mathfrak{F}_k$. We write $X^{\mathfrak{F}_k}$ for the restriction of $X_{\mathfrak{g}^*}$ to $\mathfrak{F}_k$. We also denote by $X_1^k,\ldots,X_r^k$ the restrictions of $X_1,\ldots,X_r$ to $\mathfrak{F}_k$, respectively.
	
	Since the elements of $V_{\mathfrak{g}^*}$ can be restricted to each $\mathfrak{F}_k$, such restrictions give rise to a VG Lie algebra  $V^k:=\langle X^k_1,\ldots, X^k_r\rangle$ for the system $X^{\mathfrak{F}_k}$. It may happen that $V^k$ is not isomorphic to $V_{\mathfrak{g}^*}$, since some  restrictions of the vector fields of  $V_{\mathfrak{g}^*}$ to $\mathfrak{F}_k$ may vanish. Nevertheless, the restriction mapping $\rho_k:X\in V_{\mathfrak{g}^*}\mapsto X|_{\mathfrak{F}_k}\in V^k$ is a surjective Lie algebra morphism. Using the above facts and since $V_{\mathfrak{g}^*}\simeq\mathfrak{g}$ is simple by assumption, we obtain that the kernel of $\rho_k$ must be equal to zero or $V_{\mathfrak{g}^*}$, because these are the only ideals of $V_{\mathfrak{g}^*}$. Let us prove that $\ker \rho_k=0$. Every Hamiltonian vector field on $\mathfrak{g}^*$ is of the form $X_h:=\{h,\cdot\}$ for a certain $h\in C^\infty(N)$. Since $\{e_1,\ldots,e_r\}$ are coordinates on $\mathfrak{g}^*$, one obtains that $h$ is a function of $e_1,\ldots,e_r$, i.e. $h=h(e_1,\ldots,e_r)$, and, using the properties of Poisson brackets, one gets that
 	$$
	X_hf=\{h(e_1,\ldots,e_r),f\}=\sum_{\alpha=1}^r\frac{\partial h}{\partial e_\alpha}\{e_\alpha,f\}=\sum_{\alpha=1}^r\frac{\partial h}{\partial e_\alpha}X_\alpha f,\qquad \forall f\in C^\infty(\mathfrak{g}^*).
	$$
	Hence, the vector fields $X_1,\ldots,X_r\in V_{\mathfrak{g}^*}$ span the tangent space to the symplectic leaves of $\mathfrak{g}^*$. In other words, the elements of $V^k$ span the tangent space to each $\mathfrak{F}_k$. If we assume that $\dim\mathfrak{F}_k\neq 0$, then the elements of $V^k$ span the tangent space to $\mathfrak{F}_k$ and $V^k\neq0$. Hence, $\ker \rho_k\neq V_{\mathfrak{g}^*}$ and $\ker \rho_k=0$. Thus, (\ref{g*}) can be restricted to each leaf, $\mathfrak{F}_k$, giving rise to a Lie--Hamilton system with a VG Lie algebra  of Hamiltonian vector fields $V^k\simeq \mathfrak{g}$. 
	
	Summarising, we have proved the following result.
	
	\begin{Theorem}\label{Rg*} Every simple Lie algebra $\mathfrak{g}$ defines a Lie--Hamilton system (\ref{g*}) on $\mathfrak{g}^*$ with respect  to the KKS Poisson bracket. The restriction of (\ref{g*}) to each non-zero dimensional symplectic leaf gives rise to a new Lie--Hamilton system with a VG Lie algebra isomorphic to $\mathfrak{g}$.
	\end{Theorem}
	
	In particular, if $\dim \mathfrak{F}_k=2$, then  $V^k$ is locally diffeomorphic to one of the classes of Lie--Hamilton systems on $\mathbb{R}^2$ admitting a VG Lie algebra isomorphic to $\mathfrak{g}$. 
		In view of Table \ref{table1} and since $\mathfrak{sl}^*_2$ and $\mathfrak{so}^*_3$ have two-dimensional symplectic leaves (cf. \cite{PSWZ76,SW14}), Theorem \ref{Rg*} potentially allows one to obtain the Lie--Hamilton systems on $\mathbb{R}^2$ related to the VG Lie algebras isomorphic to $\mathfrak{sl}_2$ and $\mathfrak{so}_3$, i.e.  P$_2$, P$_3$, I$_4$, and I$_5$.	According to Table \ref{table1}, there exist three classes of VG Lie algebras of Hamiltonian vector fields on $\mathbb{R}^2$ isomorphic to $\mathfrak{sl}_2$: P$_2$, I$_4$, and I$_5$. To check that Theorem \ref{Rg*} allows us to recover all of them, we have to verify all the restrictions of (\ref{g*}) on $\mathfrak{sl}^*_2$ to each symplectic leaf relative to the KKS bracket in detail (see Figure \ref{Fig1}). On the other hand, it is immediate that Theorem \ref{Rg*} gives rise to a Lie--Hamilton system on $\mathbb{R}^2$ admitting a VG Lie algebra locally diffeomorphic to P$_3$, since it is the only class of Hamiltonian VG Lie algebras on $\mathbb{R}^2$ isomorphic to $\mathfrak{so}_3$ (see Table \ref{table1}).
		
		Although Lie--Hamilton systems on $\mathbb{R}^2$ related to VG Lie algebras isomorphic to $\mathfrak{sl}_2$ and $\mathfrak{so}_3$ were very briefly studied in \cite{BCFHL18}, our analysis here is much more detailed and it additionally shows, as a bonus, the existence of additional features of such Lie--Hamilton systems, which retrieves in a more natural and general manner results given in \cite{GLMV19,LL18}.
		
\newpage
	$\bullet$ { VG Lie algebras isomorphic to $\mathfrak{sl}_2$ (classes P$_2$, I$_4$, and I$_5$)}:
	
	Let $\{e_1,e_2,e_3\}$ be a basis of $\mathfrak{sl}_2$ satisfying the commutation relations
	$$
	[e_1,e_2]=e_1,\qquad [e_1,e_3]=2e_2,\qquad [e_2,e_3]=
	e_3.
	$$
	In the given basis, system (\ref{g*}) takes the form
\begin{equation}\label{Syssl2}
\frac{de_1}{dt}=-2b_3(t)e_2-b_2(t)e_1,\quad\frac{de_2}{dt}=b_1(t)e_1-b_3(t)e_3,\quad \frac{de_3}{dt}=2b_1(t)e_2+b_2(t)e_3.
\end{equation}

\begin{figure}
    \centering
\includegraphics[scale=0.5]{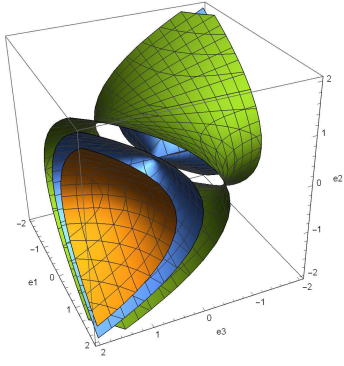}    \caption{Representative symplectic submanifolds of the KKS bracket on $\mathfrak{sl}_2^*$.}
    \label{Fig1}
\end{figure}		
	Hence, $\mathfrak{sl}_2^*$ admits a Casimir function given by
	$
	\mathfrak{C}_{\mathfrak{sl}_2}:=e_1e_3-e_2^2
	$ (see   \cite{PSWZ76,SW14}). The symplectic leaves of the KKS bracket on $\mathfrak{sl}^*_2$ are therefore given by the submanifolds where the function $\mathfrak{C}_{\mathfrak{sl}_2}$ takes a constant value. Then, Hamiltonian vector fields are tangent to such  submanifolds. This implies that we have three main types of two-dimensional symplectic leaves where $\mathfrak{C}_{\mathfrak{sl}_2}$ is positive (two-sheeted hyperboloids), negative (one-sheeted hyperboloids), or zero (a two-sided cone without vertex). Three types of symplectic leaves are illustrated in the diagram aside in orange, green, and blue colour, respectively.
	\vskip 0.1cm
	In the coordinate system
	$
	\{e_1,e_2,k:=e_1e_3-e_2^2\},
	$
	one obtains that the restriction of (\ref{Syssl2}) to each leaf of the KKS bracket, reads
	\begin{equation}\label{sl2Fk}
	\frac{de_1}{dt}=-b_2(t)e_1-2b_3(t)e_2,\qquad \frac{de_2}{dt}=b_1(t)e_1-b_3(t)\frac{k+e_2^2}{e_1},\qquad k\in \mathbb{R}.
	\end{equation}

	To determine to which class in Table \ref{table1} systems (\ref{sl2Fk}) belong to, one can use the so-called {\it Casimir tensor field} (see \cite{BHLS15})
	\begin{equation}\label{CTsl2}
	\mathfrak{T}_{\mathfrak{sl}_2}:=\frac 12(X^{k}_1\otimes X^{k}_3+X^{k}_3\otimes X^{k}_1)-X^{k}_2\otimes X^{k}_2.
	\end{equation}
	This tensor field appears naturally as the restriction to each $\mathfrak{F}_k$ of $\frac 12(X_1\otimes X_3+X_3\otimes X_1)-X_2\otimes X_2$ on $\mathfrak{sl}_2^*$. In fact, since the previous tensor field is constructed as a linear combination of tensor products of Hamiltonian vector fields, it is then tangent to each $\mathfrak{F}_k$ and it gives rise to $\mathfrak{T}_{\mathfrak{sl}_2}$. 
	
	Recall that the two-dimensional symplectic leaves in $\mathfrak{sl}_2^*$ relative to the KKS bracket are given by its symplectic leaves in  $\mathfrak{sl}^*_2\backslash\{0\}$. Hence, one can use  \cite[Theorem 4.4]{BHLS15}, which claims that the sign, $k_{\mathfrak{T}_{\mathfrak{sl}_2}}$, of the matrix of the coefficients of  $\mathfrak{T}_{\mathfrak{sl}_2}$ at a single point of $\mathfrak{F}_k$ classifies to which class of Table \ref{table1}, the Lie algebra $V^k\simeq \mathfrak{sl}_2$ is locally diffeomorphic to. In particular, if $k_{\mathfrak{T}_{\mathfrak{sl}_2}}>0$, then $V^k$ is locally diffeomorphic to P$_2$; while $k_{\mathfrak{T}_{\mathfrak{sl}_2}}=0$ gives that $V^k$ is locally diffeomorphic to I$_5$; and finally $k_{\mathfrak{T}_{\mathfrak{sl}_2}}<0$ indicates that $V^k$ is locally diffeomorphic to I$_4$. 
	
	To check which class of Lie algebras on the plane is $V^k$ diffeomorphic to, it is enough to see that $\{e_1,e_2\}$ are local coordinates at a certain open subset of each two-dimensional leaf $\mathfrak{F}_k$. Hence, $k_{\mathfrak{T}_{\mathfrak{sl}_2}}={\rm sign}(\det [\mathfrak{T}_{\mathfrak{sl}_2}]_{\alpha\beta})={\rm sign}(\det \mathfrak{T}_{\mathfrak{sl}_2}(de_\alpha,de_\beta))$ for $\alpha,\beta=1,2$. Then,
	\begin{equation}\label{class}
	\det(\mathfrak{T}_{\mathfrak{sl}_2})=\det\left[\begin{array}{cc}-e_1^2&-e_2e_1\\-e_2e_1&-e_1e_3\end{array}\right]=e_1^2(e_1e_3-e_2^2)=e_1^2k.
	\end{equation}
Thus, each one of the main three types of symplectic leaves leads to a VG Lie algebra of a different type. In particular, $k>0$ yields a VG Lie algebra locally diffeomorphic to P$_2$ on a two-sheeted hyperboloid; $k<0$ leads to a VG Lie algebra locally diffeomorphic to I$_4$ on a one-sheeted hyperboloid; and $k=0$ induces a VG Lie algebra locally diffeomorphic to I$_5$ on a two-sided cone without vertex. In turn, this makes the Lie system
\begin{equation}\label{restr}
X^{\mathfrak{F}_k}=\sum_{\alpha=1}^3b_\alpha(t)X_\alpha^k
\end{equation}
into a Lie--Hamilton system possessing a VG Lie algebra locally diffeomorphic to P$_2$, I$_4$, and I$_5$, when $k>0$, $k<0$, or $k=0$, respectively. Recall that the symplectic form related to (\ref{restr}) is the one associated with the restriction to $\mathfrak{F}_k$ of the Poisson bivector on $\mathfrak{sl}^*_2$.

For the sake of completeness, we provide the coordinate expression for (\ref{restr}). In particular the system can be mapped, for $k>0$, onto the system
	$$
	\frac{d\xi}{dt}=b_1(t)\frac{\partial}{\partial x}+b_2(t)\left(x\frac{\partial}{\partial x}+y\frac{\partial}{\partial y}\right)+b_3(t)\left((x^2-y^2)\frac{\partial}{\partial x}+2xy\frac{\partial}{\partial y}\right), \qquad \xi:=(x,y)\in \mathbb{R}^2,
	$$
	through the change of variables $x:=e_2/e_1$ and $y:=\sqrt{k}/e_1$. In the case $k<0$, system (\ref{sl2Fk}) can be mapped, via the change of variables $x:=(e_2-\sqrt{-k})/e_1$ and $y:=(e_2+\sqrt{-k})/e_1$, onto the system
	$$
	\frac{d\xi}{dt}=b_1(t)\left(\frac{\partial}{\partial x}+\frac{\partial}{\partial y}\right)+b_2(t)\left(x\frac{\partial}{\partial x}+y\frac{\partial}{\partial y}\right)+b_3(t)\left(x^2\frac{\partial}{\partial x}+y^2\frac{\partial}{\partial y}\right),\qquad \xi:=(x,y)\in \mathbb{R}^2.
	$$
	Finally, the change of variables $x:=e_2/e_1, y:=1/\sqrt{|e_1|}$ maps (\ref{sl2Fk}) into
	$$
	\frac{d\xi}{dt}=b_1(t)\frac{\partial}{\partial x}+b_2(t)\left(x\frac{\partial}{\partial x}+\frac y2\frac{\partial}{\partial y}\right)+b_3(t)\left(x^2\frac{\partial}{\partial x}+xy\frac{\partial}{\partial y}\right),\qquad \xi:=(x,y)\in \mathbb{R}^2.
	$$
	
	It is worth noting that our approach also explains other aspects of the  Lie--Hamilton systems (\ref{restr}). In fact, the Casimir tensor field (\ref{CTsl2}) can also be understood as a symmetric contravariant tensor field on each leaf $\mathfrak{F}_k$ that is invariant relative to the Lie derivative with the elements of $V^k$. When $k\neq 0$, the matrix of the coefficients of (\ref{CTsl2}) becomes non-degenerate by (\ref{class}), and $\mathfrak{T}_{\mathfrak{sl}_2}$ admits an inverse $g_{\mathfrak{sl}_2}$ that is a metric invariant relative to the Lie derivaties with the elements of $V^k$. Then, the Lie--Hamilton systems on the leaves with $k\neq 0$ become  pseudo-Riemannian Lie systems as proved in \cite[Table 1]{LL18}. In particular, the metric is positive  definite for $k>0$, i.e. Riemannian, while it is pseudo-Riemannian for $k<0$.
	
		$\bullet$ { VG Lie algebra isomorphic to $\mathfrak{so}_3$ (class P$_3$):}
	
	Consider the basis of $\mathfrak{so}_3$ of the form $\{e_1,e_2,e_3\}$ admitting the structure constants\footnote{We recall that these commutation relations are not the same as the ones of the basis of P$_3$ given in Table \ref{table1}.}
 	$$
 	[e_1,e_2]=e_3,\qquad [e_2,e_3]=e_1,\qquad [e_3,e_1]=e_2.
	$$
Then, system (\ref{g*}) takes in this case the form
\begin{equation}\label{so3}
\frac{de_1}{dt}=b_3(t)e_2-b_2(t)e_3,\quad\frac{de_2}{dt}=b_1(t)e_3-b_3(t)e_1,\quad\frac{de_3}{dt}=b_2(t)e_1-b_1(t)e_2.
\end{equation}
	In this case, the Casimir function of the KKS Poisson structure on $\mathfrak{so}_3^*$ reads \cite{PSWZ76,SW14}
$$
\mathfrak{C}_{\mathfrak{so}_3}:=e_1^2+e_2^2+e_3^2.
$$
	The two-dimensional submanifolds where $\mathfrak{C}_{\mathfrak{so}_3}$ becomes a constant $k^2$, i.e. the symplectic surfaces associated with the KKS bracket, are therefore spheres, $\mathfrak{F}_k$, with different radius $k>0$. As previously explained in the general theory of this section, the Lie algebra $V^k$ defined on each sphere becomes a VG Lie algebra isomorphic to $\mathfrak{so}_3$. In view of Table \ref{table1}, the Lie algebra $V^k$ is locally diffeomorphic to  P$_3$.
	
	For the sake of completeness, we can describe the restriction process in detail. If one consider the set of variables 
	$$
	k:=\sqrt{e_1^2+e_2^2+e_3^2}, \qquad r:=\sqrt{e_1^2+e_2^2},\qquad {\rm tg}\,\varphi:=\frac{e_2}{e_1}, 
	$$
	the restriction of (\ref{so3}) to a symplectic leaf $\mathfrak{F}_k$ reads
	\begin{equation}\label{Resso3}
	\frac{dr}{dt}=\sqrt{k^2-r^2}[b_1(t)\sin\varphi -b_2(t)\cos\varphi],\qquad \frac{d\varphi}{dt}=-b_3(t)+\sqrt{k^2/r^2-1}[b_1(t)\cos\varphi+b_2(t)\sin\varphi].
	\end{equation}
As stated previously, there must exist a change of variables mapping this system of differential equations into a Lie system with the VG Lie algebra P$_3$. In fact, the change of variables
 $$
 x:= -\frac{1}{r}[\sqrt{k^2-r^2}+ k]\sin \varphi,\qquad y:=\frac 1r[\sqrt{k^2-r^2}+ k] \cos \varphi
 $$
 maps (\ref{Resso3}), i.e. the restriction of (\ref{so3}) to $\mathfrak{F}_k$,  onto 
 $$
 \frac{d\xi}{dt}=b_1(t)\left(\frac 12(1+x^2-y^2)\frac{\partial}{\partial x}+xy\frac{\partial}{\partial y}\right)+b_2(t)\left(xy\frac{\partial}{\partial x}+\frac 12(1+y^2-x^2)\frac{\partial}{\partial y}\right)+b_3(t)\left(y\frac{\partial}{\partial x}-x\frac{\partial}{\partial y}\right),
 $$
 where $\xi:=(x,y)\in \mathbb{R}^2$.
 
It is worth noting that the Casimir tensor field relative to $\mathfrak{C}_{\mathfrak{so}_3}$, namely
\begin{equation}\label{Rie}
\mathfrak{T}_{\mathfrak{so}_3}:=X_1^{k}\otimes X_1^{k}+X_2^{k}\otimes X_2^{k}+X_3^{k}\otimes X_3^{k},
\end{equation}
is tangent to $\mathfrak{F}_k$. The coefficients of $\mathfrak{T}_{\mathfrak{so}_3}$ in the local coordinates $\{e_1,e_2\}$, which are well defined on a point with $k^2> e_1^2+e_2^2$ of every leaf $\mathfrak{F}_k$ for $k>0$, read
	\begin{equation}\label{metric2}
	\det(\mathfrak{T}_{\mathfrak{so}_3})=\det\left[\begin{array}{cc}e_3^2+e_2^2&-e_1e_2\\-e_2e_1&e_3^2+e_1^2\end{array}\right]=(k^2-e_1^2-e_2^2)k^2> 0
	\end{equation}
and $\mathfrak{T}_{\mathfrak{so}_3}$ becomes a non-degenerate tensor field. Its inverse, $g_{\mathfrak{so}_3}$, becomes a Riemmanian metric since the determinant of its coefficients is positive. As a similar result can be obtained in any  local coordinate system around points of each $\mathfrak{F}_k$, the leaf $\mathfrak{F}_k$ becomes a Riemmanian manifold. Since (\ref{Rie}) satisfies that $\mathcal{L}_{X^k_i}\mathfrak{T}_{\mathfrak{so}_3}=0$ for $i=1,2,3$, then  $V^k$ consists of Killing vector fields with respect to $g_{\mathfrak{so}_3}$, as proved in \cite{LL18}, and the Lie--Hamilton systems on $\mathfrak{F}_k$ are also Riemannian--Lie systems.

To finish this section, we stress that Theorem \ref{Rg*} implies that the Lie algebra $\mathfrak{sl}_3$ or other simple Lie algebras have no two-dimensional symplectic leaves. Otherwise, such Lie algebras of vector fields would give rise to Lie algebras of Hamiltonian vector fields on $\mathbb{R}^2$, which are absent in Table \ref{table1}. Moreover, the method exposed here can be employed to obtain Lie--Hamilton systems on higher-dimensional manifolds. 

	\section{Non-simple VG Lie algebras and symplectic foliations}\label{NonSimple}
It is immediate that the method showed in the last section can be slightly modified to give rise to Lie--Hamilton systems on the non-zero-dimensional symplectic leaves of the symplectic stratifications of the KKS bracket of general, not necessarily simple, Lie algebras $\mathfrak{g}$. Nevertheless, the induced VG Lie algebras on each leaf, $V^k$, do not need to be isomorphic neither to $V_{\mathfrak{g}^*}$ nor to $\mathfrak{g}$. The difficulties due  to the lack of an isomorphism between $\mathfrak{g}$ and $V^k$ were not addressed in \cite{BCFHL18}. Despite above mentioned problems, we prove in this section that the method given in Section \ref{LHSg*} can be extended to the Lie algebras $\mathfrak{iso}_2$ and $\mathfrak{iso}_{1,1}$ to obtain  Lie--Hamilton systems on their symplectic leaves related to VG Lie algebras locally diffeomorphic to elements of the class I$_{14A}$  for $r=2$. 
	
	$\bullet$ VG Lie algebra isomorphic to $\mathfrak{iso}_2$ (class I$_{14A}$ for $r=2$):

Let $\{e_1,e_2,e_3\}$ be a basis of the Lie algebra $\mathfrak{iso}_2$, i.e. the abstract Lie algebra isomorphic to a VG Lie algebra of the class I$_{14A}$, with commutation relations
\begin{equation}\label{ConRel}
[e_1,e_2]=0,\qquad [e_1,e_3]=-e_2,\qquad [e_2,e_3]=e_1.
\end{equation}
The Poisson structure on $\mathfrak{iso}_2^*$ given by the KKS bracket admits a Casimir function \cite{PSWZ76,SW14} given by
$$
\mathfrak{C}_{\mathfrak{iso}_2}:=e_1^2+e_2^2.
$$
Consequently, the symplectic stratification on $\mathfrak{iso}_2^*$ is given by the cylinders $e_1^2+e_2^2=k$, with $k>0$, and the set of points with $e_1=e_2=0$. We focus on the two-dimensional symplectic leaves, $\mathfrak{F}_k$ with $k>0$, of this stratification.
Using (\ref{ConRel}), we obtain that the vector fields $X_\alpha:=\{e_\alpha,\cdot\}$ for $\alpha=1,2,3$ read
$$
X_1=-e_2\frac{\partial}{\partial e_3},\qquad X_2=e_1\frac{\partial}{\partial e_3},\qquad X_3=e_2\frac{\partial}{\partial e_1}-e_1\frac{\partial}{\partial e_2},
$$
on the coordinates $\{e_1,e_2,e_3\}$ of $\mathfrak{iso}^*_2$.
The restrictions of $X_1,X_2,X_3$ to the two-dimensional symplectic leaves $\mathfrak{F}_k$ can be expressed on each leaf by restricting the cylindrical coordinates $\{r:=\sqrt{e_1^2+e_2^2}=\sqrt{k},$ $\varphi:={\rm arc\, tan}(e_2/e_1), e_3\}$ on $\mathfrak{iso}^*_2$ to each leaf, i.e. via $\{\varphi,e_3\}$. Then,
$$
X^k_1=-r\sin \varphi\frac{\partial}{\partial e_3},\qquad X^k_2=r\cos \varphi\frac{\partial}{\partial e_3},\qquad X^k_3=-\frac{\partial}{\partial \varphi}.
$$
Therefore, one obtains that $X^k_1,X^k_2,X^k_3$ are Hamiltonian and span a three-dimensional Lie algebra $V^k\simeq \mathfrak{iso}_2$. Since $\langle X^k_1,X^k_2\rangle \simeq \mathbb{R}^2$, $X_1^k\wedge X_2^k=0$, and the Lie brackets of $X_3^k$ with elements of $\langle X^k_1,X^k_2\rangle$ belong also to  the latter linear space, $V^k$ must be locally diffeomorphic to a Lie algebra of the class I$_{14A}$ for $r=2$.

$\bullet$ VG Lie algebras isomorphic to $\mathfrak{iso}_{1,1}$ (class I$_{14A}$ for $r=2$):

Let $\{e_1,e_2,e_3\}$ be a basis  of $\mathfrak{iso}_{1,1}$  with the structure constants given by
$$
[e_1,e_2]=0,\qquad [e_1,e_3]=e_1,\qquad [e_2,e_3]=-e_2.
$$
Then, a Casimir function is given by 
$$
\mathfrak{C}_{\mathfrak{iso}_{1,1}}:=e_1e_2.
$$
Hence, the two-dimensional symplectic leaves associated with the KKS Poisson bracket on $\mathfrak{iso}_{1,1}$ are given by surfaces $e_1e_2=k\neq 0$ or the four  semiplanes with $e_1=0$ or $e_2=0$. We now  write $\mathfrak{F}_k$ for the symplectic leaf for a certain $k\neq 0$.

Let us derive the restrictions of the vector fields $X_1,X_3,X_3$ on $\mathfrak{iso}_{1,1}^*$, i.e. 
$$
X_1=e_1\frac{\partial}{\partial e_3},\qquad X_2=-e_2\frac{\partial}{\partial e_3},\qquad X_3=-e_1\frac{\partial}{\partial e_1}+e_2\frac{\partial}{\partial e_2},
$$
 to the symplectic leaves $\mathfrak{F}_k$ by defining the coordinates $e_1=:e^{h},k:=e_1e_2$, and $e_3$, when they will have sense. This gives rise to the expressions
$$
X^k_1=e^h\frac{\partial}{\partial e_3},\qquad X^k_2=-e^{-h}k\frac{\partial}{\partial e_3},\qquad X^k_3=-\frac{\partial}{\partial h}.
$$
Therefore, the restriction to a leaf $\mathfrak{F}_k$ with $k\neq 0$ is a Lie system with a VG Lie algebra belonging to the class I$_{14A}$ for $r=2$.

It is immediate that $V^k$ for $k=0$ gives rise to a VG Lie algebra on a semi-plane isomorphic to $\mathfrak{h}_{2}$, i.e. a non-Abelian two-dimensional Lie algebra, and locally diffeomorphic to I$_{14A}$ for $r=1$.

\section{An approach through automorphic Lie systems and left-invariant bivector fields}

The method established in Sections \ref{LHSg*} and \ref{NonSimple} cannot be easily applied to all Lie--Hamilton systems on $\mathbb{R}^2$. On the one hand, this is due to the fact that the abstract Lie algebras isomorphic  to the VG Lie algebras in Table \ref{table1} may not posses a sufficient number of Casimir functions or invariants to reduce the system to a two-dimensional submanifold with a compatible symplectic structure, e.g. the Lie algebra I$_{16}$ for $r=1$ has no Casimir function (cf. \cite{PSWZ76,SW14}). On the other hand, one could consider other Lie algebras $\mathfrak{g}$ not given in Table \ref{table1} and expect that the restriction of the vector fields $X_1,\ldots,X_r$ on $\mathfrak{g}^*$ onto certain two-dimensional symplectic leaves (relative to the KKS bracket) will give rise to Lie--Hamilton systems with the remaining VG Lie algebras in Table \ref{table1}. But there is no hint about the form of the general Lie algebras $\mathfrak{g}$ to be used. 

Above problems motivate our following alternative model for Lie--Hamilton systems on $\mathbb{R}^2$ (and other possible ones on higher-dimensional manifolds). Our approach suggests that Lie--Hamilton systems on $\mathbb{R}^2$ can be recovered as projections of automorphic Lie systems with appropriate projectable invariants in $\mathcal{T}_{\rm inv}^L(G)$ or, more generally, elements of $\mathcal{T}(G)$.

Let us provide our new approach. Consider a Lie system $X=\sum_{\alpha=1}^rb_\alpha(t)X_\alpha$ on a two-dimensional connected manifold $N$ with a VG Lie algebra $V:=\langle X_1,\ldots,X_r\rangle$. The integration of $V$ allows us to obtain a Lie group action $\Phi:G\times N\rightarrow N$ whose Lie algebra of fundamental vector fields is equal to $V$ and $\dim G=\dim V$. Assume that $\mathcal{D}^V=TN$, which yields that $V$ is a locally transitive VG Lie algebra. Then, 
$$
\Phi_x:g\in G\mapsto \Phi(g,x)\in N
$$
is surjective for every $x\in N$. If $G_x$ is the isotropy group of $x\in N$, then one gets the commutative diagram
\begin{equation}\label{diag}
\begin{tikzpicture}[
roundnode/.style={},
squarednode/.style={},node distance=1.5cm,auto
]
\node(G)   {$G$};
\node(N)  [right=of G] {$N$};
\node(quotient) [below=of G] {$G/G_x$};

\draw[->] (G)  edge node {$\Phi_x$} (N) ;
\draw[->] (G) edge node {$\pi$} (quotient);
\draw[->] (quotient) edge node [below] {$\widehat{\Phi}_x$} (N);

\end{tikzpicture}
\end{equation}
where $\widehat {\Phi}_x$ is a diffeomorphism. Moreover, if $X^R_v$ is the right-invariant vector field on $G$ with $X^R_v(0)=v$ for $v\in \mathfrak{g}$ and $X_v$ is the fundamental vector field\footnote{We choose the fundamental vector field $X_v$ to be defined as $X_v(x):=\frac{d}{dt}\big|_{t=0}\Phi(\exp(-tv),x)$ for every $x\in N$} of $\Phi$ associated with $v$, then
$$
\Phi_{x*g}({X_v^R})_g=\frac{d}{dt}\bigg|_{t=0}\!\!\!\Phi_x(\exp(tv)g)=-X_v(gx),\,\, \forall g\in G\,\, \Rightarrow\,\, \Phi_{x*}X^R_v=-X_v,\,\, \forall v\in \mathfrak{g}.
$$
In particular,  $\Phi_{x*}X^R_\alpha=-X_\alpha$, with $\alpha=1,\ldots, r$, for a certain basis $\{X_1^R,\ldots,X_r^R\}$ of $V^R$. 
Due to the commutativity of the diagram (\ref{diag}), one obtains that
$
\pi_*(V^R)\simeq  V,
$
where $V^R$, as usually, stands for the Lie algebra of right-invariant vector fields on $G$. 

Our aim now is to obtain a Poisson bivector $\Lambda$ on $G$ that can be projected onto $G/G_x$, e.g. $\Lambda$ is a left-invariant bivector field that is invariant relative to the left-invariant vector fields tangent to $G_x$. This shows that $\pi_*\Lambda$ exists.  Then, $\pi_*\Lambda$ induces a Poisson bivector on $G/G_x$ that is invariant relative to $\pi_*(V^R)$. Hence, this allows us to define a Lie--Hamilton system on $G/G_x$ given by 
$X^\pi:=-\sum_{\alpha=1}^rb_\alpha (t)\pi_*X^R_\alpha$, which is diffeomorphic to the one, $X$, on $N$ via $\widehat{\Phi}_x$. If $\pi_*\Lambda$ does not vanish, then $X$ is a Lie--Hamilton system relative to a symplectic structure.

To illustrate the above method, let us analyse Lie--Hamilton systems related to P$_5$, I$_{16}$, P$_1$, and I$_8$.

$\bullet$ VG Lie algebra isomorphic $\mathfrak{sl}_2\ltimes \mathbb{R}^2$ (class P$_5$):

We consider the Lie group $SL_2\ltimes \mathbb{R}^2$ as the direct product of the matrix Lie group $SL_2$ of $2\times 2$ unimodular matrices with $\mathbb{R}^2$ relative to the multiplication
$$
\left(A,\vec{r}\right)\star\left(B,\vec{s}\right):=
\left(AB,A\vec{s}+\vec{r}\right),\qquad \forall A,B\in SL_2,\forall \vec{r},\vec{s}\in \mathbb{R}^2.
$$
Let us analyse now the automorphic Lie system on $SL_2\ltimes \mathbb{R}^2$  given by
$$
X^R(t,g):=-\sum_{\alpha=1}^5b_\alpha(t)X^R_\alpha(g),\qquad \forall g\in SL_2\ltimes \mathbb{R}^2,
$$
for a certain basis $X^R_1,\ldots,X_5^R$ of right-invariant vector fields on $SL_2\ltimes \mathbb{R}^2$. 
Let us use a local coordinate system $\{\alpha,\beta, \gamma,\sigma,\epsilon\}$ defined close to the neutral element of $SL_2\ltimes \mathbb{R}^2$ and related to the description of elements of $SL_2\ltimes \mathbb{R}^2$ in the form
$$
\left(\left[\begin{array}{cc}\alpha&\beta\\\gamma&\delta\end{array}\right], \left[\begin{array}{c}\sigma\\\epsilon\end{array}\right]\right),\qquad \left[\begin{array}{cc}\alpha&\beta\\\gamma&\delta\end{array}\right]\in SL_2,\quad \left[\begin{array}{c}\sigma\\\epsilon\end{array}\right]\in \mathbb{R}^2.
$$
In this coordinate system, a basis of left-invariant vector fields on $SL_2\times \mathbb{R}^2$ reads
$$
\begin{gathered}
X^L_1=\alpha\frac{\partial}{\partial \alpha}-\beta\frac{\partial}{\partial \beta}+\gamma\frac{\partial}{\partial \gamma},\qquad X_2^L=\alpha\frac{\partial}{\partial \beta},\qquad X_3^L=\beta\frac{\partial}{\partial \alpha}+\frac{1+\beta\gamma}{\alpha}\frac{\partial}{\partial \gamma},\\
X_4^L=\alpha\frac{\partial}{\partial \sigma}+\gamma\frac{\partial}{\partial \epsilon},\qquad X_5^L=\beta\frac{\partial}{\partial \sigma}+\frac{1+\beta\gamma}{\alpha}\frac{\partial}{\partial \epsilon}.
\end{gathered}
$$
They span indeed a VG Lie algebra isomorphic to the VG Lie algebras of the class P$_5$.
Meanwhile,
\begin{equation}\label{basis}
\begin{gathered}
X_1^R=\alpha\frac{\partial}{\partial \alpha}+\beta\frac{\partial}{\partial \beta}-\gamma\frac{\partial}{\partial \gamma}+\sigma\frac{\partial}{\partial \sigma}-\epsilon\frac{\partial}{\partial \epsilon},\quad X^R_2=\gamma\frac{\partial}{\partial \alpha}+\frac{1+\beta\gamma}{\alpha}\frac{\partial}{\partial \beta}+\epsilon\frac{\partial}{\partial \sigma},\\
X_3^R=\beta\frac{\partial}{\partial \alpha}+\frac{1+\beta\gamma}{\alpha}\frac{\partial}{\partial \gamma},\qquad
X_4^R=\frac{\partial}{\partial \sigma},\qquad X_5^R=\frac{\partial}{\partial \epsilon}.
\end{gathered}
\end{equation}
Since  $\Lambda:=X_4^L\wedge X_5^L$ is a exterior product of left-invariant vector fields, it is then invariant relative to the Lie derivatives with respect to the elements of $V^R$. 

If we consider the quotient space $(SL_2\ltimes \mathbb{R}^2)/SL_2$ of left-cosets of $SL_2$, we obtain the projection $\pi:SL_2\ltimes \mathbb{R}^2\rightarrow (SL_2\ltimes \mathbb{R}^2)/SL_2$ which satisfies that $\ker \pi_{*g}=\langle (X^L_1)_g,(X_2^L)_g,(X^L_3)_g\rangle$. Since $\Lambda$ and all the elements of $V^R$ are invariant with respect to the Lie derivatives relative to the elements of $\langle X^L_1,X_2^L,X_3^L\rangle$, one obtains that $V^R$ and $\Lambda$ are projectable onto $(SL_2\ltimes \mathbb{R}^2)/SL_2$. The projections of the elements of $V^R$ onto $(SL_2\ltimes \mathbb{R}^2)/SL_2$ are spanned by linear combination over the reals of the elements of the basis
$$
\begin{gathered}
X_1^\pi=\sigma\frac{\partial}{\partial \sigma}-\epsilon\frac{\partial}{\partial \epsilon},\qquad X^\pi_2=\epsilon\frac{\partial}{\partial \sigma},\qquad
X_3^\pi=\delta\frac{\partial}{\partial \gamma},\qquad
X_4^R=\frac{\partial}{\partial \sigma},\qquad X_5^R=\frac{\partial}{\partial \epsilon}.
\end{gathered}
$$
Therefore, the projection of the vector fields (\ref{basis}) onto $SL_2\ltimes \mathbb{R}^2$ span, as proved in our general theory, a VG Lie algebra isomorphic to $\mathfrak{sl}_2\ltimes \mathbb{R}^2$. Moreover, the projection of $X^R$ onto $(SL_2\ltimes \mathbb{R}^2)/SL_2$ takes the form
\begin{equation}\label{Proj}
X^\pi(s):=-\sum_{\alpha=1}^5b_\alpha(t)X_\alpha^\pi(s),\qquad \forall s\in (SL_2\ltimes \mathbb{R}^2)/SL_2,
\end{equation}
and $X^\pi$ admits a VG Lie algebra $V^\pi=\langle X_1^\pi,\ldots,X_5^\pi\rangle\simeq \mathfrak{sl}_2\ltimes \mathbb{R}^2$.
Moreover, one has that the projection, $\Lambda^\pi$, of $\Lambda$ onto $SL_2\ltimes \mathbb{R}^2$ reads
$$
\Lambda^\pi=\frac{\partial}{\partial \sigma}\wedge \frac{\partial}{\partial \epsilon},
$$
which becomes equivalent to a symplectic form. Since $\Lambda$ is invariant relative to the elements of $V^R$, it follows that $\Lambda^\pi$ is invariant with respect to the elements of $V^\pi$, and (\ref{Proj}) becomes a Lie--Hamilton system relative to $\Lambda^\pi$. 
Recall that the integration of the VG Lie algebra P$_5$ gives rise to a Lie group action $\Phi:(SL_2\ltimes \mathbb{R}^2)\times \mathbb{R}^2\rightarrow \mathbb{R}^2$. In view of the basis of P$_5$ in Table \ref{table1}, one sees that $SL_2$ can be considered as the isotropy group of $0\in \mathbb{R}^2$. Then, the 
mapping $\Phi_0:SL_2\times \mathbb{R}^2\rightarrow \mathbb{R}^2$ gives rise to a diffeomorphism $\widehat {\Phi}_0:(SL_2\ltimes \mathbb{R}^2)/SL_2\rightarrow \mathbb{R}^2$ that maps (\ref{Proj}) onto a Lie--Hamilton system on $\mathbb{R}^2$ of the form
$$
X=-\sum_{\alpha=1}^5b_\alpha(t)\widehat{\Phi}_{0*}X^\pi_\alpha,
$$
where $\{\widehat{\Phi}_{0*}X^\pi_\alpha,\ldots,\widehat{\Phi}_{0*}X^\pi_\alpha \}$ is a basis of P$_5$. 

$\bullet$ VG Lie algebra isomorphic to  $\mathfrak{h}_2\ltimes \mathbb{R}^{r+1}$ (class I$_{16}$):

The Lie algebra, $\mathfrak{h}_2\ltimes \mathbb{R}^{r+1}$, which is isomorphic to the VG Lie algebra I$_{16}$ for the fixed $r$, has an associated Lie group given by $H_2\ltimes \mathbb{R}^{r+1}$. This Lie group acts on $\mathbb{R}^2$ having a set of fundamental vector fields given in the class I$_{16}$ of Table \ref{table1}. The Lie algebra of vector fields vanishing at $0\in \mathbb{R}^3$ is given by $\langle x\partial_y-y\partial_x,x\partial_ y,\ldots, x^r\partial_y\rangle$. Then, the isotropy group at $0\in \mathbb{R}^2$ is $\mathbb{R}\ltimes \mathbb{R}^{r}$, and its corresponding Lie algebra is isomorphic to $\mathbb{R}\ltimes \mathbb{R}^{r}$.

Let us  choose a basis of $\mathfrak{h}_2\ltimes \mathbb{R}^{r+1}$ of the form $\{e_1,\ldots,e_{r+3}\}$ obeying the same commutation relations as the basis of I$_{16}$ in Table \ref{table1}. In other words, the chosen basis has  non-vanishing commutation relations
$$
[e_1,e_3]=e_1,\,\, [e_1,e_4]=e_2,\,\, [e_1,e_{4+i}]=(i+1)e_{3+i},\,\, [e_2,e_3]=-e_2,\,\, [e_3,e_{3+i}]=(1+r)e_{3+i},
$$
for $i=1,\ldots,r$. If $\{X^L_1,\ldots, X^L_{r+3}\}$ is a basis of left-invariant vector fields on $H_2\ltimes \mathbb{R}^{r+1}$ such that $X^L_\alpha(e)=e_\alpha$ for $\alpha=1,\ldots,r+3$, one can consider the bivector field on $H_2\ltimes \mathbb{R}^{r+1}$ of the form
$$
\Lambda:=X^L_1\wedge X^L_2.
$$
This bivector field is invariant relative to the Lie derivatives with right-invariant vector fields. Moreover, $\Lambda$ can be projected onto $(H_2\ltimes \mathbb{R}^{r+1})/(\mathbb{R}\ltimes \mathbb{R}^{r})$ together with the automorphic Lie system on $H_2\ltimes \mathbb{R}^{r+1}$ of the form
$$
X^R(t,g):=-\sum_{\alpha=1}^{r+3}b_\alpha(t)X^R_\alpha(g),\qquad \forall g\in H_2\times \mathbb{R}^{r+1}.
$$
If $\pi:H_2\ltimes \mathbb{R}^{r+1}\rightarrow (H_2\ltimes \mathbb{R}^{r+1})/(\mathbb{R}\ltimes\mathbb{R}^{r})$ is the quotient mapping, $\pi_*\Lambda$ is a Poisson bivector that is invariant relative to $X^\pi_\alpha:=\pi_*X^R_\alpha$, for $\alpha=1,\ldots,5$, and this leads to a Lie--Hamilton system on $(H_2\ltimes \mathbb{R}^{r+1})/(\mathbb{R}\ltimes\mathbb{R}^{r})$ given by 
$$
X^\pi:=-\sum_{\alpha=1}^{r+3}b_\alpha(t)X^\pi_\alpha.
$$
In view of our previous general comments, 
the VG Lie algebra of the projection is locally diffeomorphic to I$_{16}$, and it represents geometrically this class of Lie--Hamilton systems on $\mathbb{R}^2$.

$\bullet$ VG Lie algebras P$_1$ and I$_8$:

These cases follow exactly the same ideas given in previous examples, hence we will just sketch the procedure. In the case of the VG Lie algebra P$_1$, its vector fields can be integrated to define an action of the Lie group $\mathbb{R}\ltimes \mathbb{R}^2$ on $\mathbb{R}^2$ admitting a set of fundamental vector fields given by P$_1$. One then considers an automorphic Lie system $X^R=-\sum_{\alpha=1}^3b_\alpha(t)X^R_\alpha$ on the Lie group $\mathbb{R}\ltimes \mathbb{R}^2$, where $X^R_1,X^R_2,X^R_3$ are assumed to admit the same commutation relations as the basis $X_1,X_2,X_3$ of P$_1$ given in Table \ref{table1}. The isotropy group of the point $0\in\mathbb{R}^2$ is given by a Lie subgroup $H\simeq \mathbb{R}$ whose Lie algebra is spanned by $X^L_3$.

Then, $\Lambda:=X_1^L\wedge X_2^L$ is a Poisson bivector on the Lie group $\mathbb{R}\ltimes \mathbb{R}^2$ that is also invariant relative to the Lie derivatives with respect to the VG Lie algebra $\langle X_1^R,X_2^R,X_3^R\rangle$. Thus, $\Lambda$ and $X^R$ are invariant relative to $X^L_3$ and they can be therefore projected simultaneously onto $(\mathbb{R}\ltimes \mathbb{R}^2)/H$ via $\pi$, giving rise to a Lie--Hamilton system on the plane  relative to $\pi_*\Lambda$ with a VG Lie algebra locally diffeomorphic to P$_1$. In the case of I$_8$, the procedure is absolutely analogous.

Let us provide a slight generalisation of the above procedure that will allow us to retrieve the Lie--Hamilton systems on the plane admitting a VG Lie algebra locally diffeomorphic to any of the VG Lie algebras of the classes I$_{14A}$ and I$_{14B}$. 

$\bullet$ VG Lie algebras isomorphic to  $\mathbb{R}\ltimes \mathbb{R}^r$ (classes I$_{14A}$ and I$_{14B}$):

Any VG Lie algebra $V$ of the classes I$_{14A}$ or I$_{14B}$ can be integrated to give rise to a Lie group action of a Lie group $G_r:=\mathbb{R}\ltimes \mathbb{R}^r$ on $\mathbb{R}^2$ whose fundamental vector fields are given by $V$. The Lie algebra of the isotropy group of $0$ is given by an ${(r-1)}$-dimensional Lie subalgebra $V_r$ of $\langle X_2,\ldots,X_r\rangle$, where we can assume without loss of generality that  $X_2$ does not vanish at $0$ and therefore $X_2\notin V_r$. Thus, $V_r$ is isomorphic to the Abelian Lie algebra $\mathbb{R}^{r-1}$ with a Lie subgroup $H$. Then, $\mathbb{R}^2$ is diffeomorphic to the quotient space $G_r/H$ with the Abelian Lie group $H\simeq \mathbb{R}^{r-1}$. We write $V^L_r$ for the Lie algebra of the left-invariant vector fields of the isotropy group of $0$, which is isomorphic to $V_r$.

We define the automorphic Lie system on  $G_r$ of the form
\begin{equation}\label{sysI14}
X^R(t,g):=-\sum_{\alpha=1}^{r+1}b_\alpha(t)X^R_\alpha,\qquad \forall g\in G_r,
\end{equation}
where $X^R_1,\ldots,X_{r+1}^R$ are assumed to close the opposite  structure constants than the basis $X_1,\ldots,X_r$ of the VG Lie algebra on $\mathbb{R}^2$ under inspection. Let us also define $J:=X_1^L\wedge X_2^L$, which is not necessarily a Poisson bivector as\footnote{See \cite{Va94} for our convention on the Schouten-Nijenhuis bracket.} $[J,J]_{SN}=-2[X_1^L,X_2^L]\wedge X_1^L\wedge X_2^L$ and $[X_1^L,X_2^L]$, which may different for the VG Lie algebras of the classes I$_{14A}$ and I$_{14B}$. 

Let us try to project $X^R$ and $J$ onto the two dimensional quotient space $G_r/H$. Since $X^R$ is invariant relative to left-invariant vector fields, it will be projectable onto $G_r/H$ via the canonical projection $\pi$. Meanwhile, for every $X^L\in V^L_r$, one gets that  $\mathcal{L}_{X^L}J=\mathcal{L}_{X^L}X^L_1\wedge X_2^L=[X^L,X_1^L]\wedge X_2^L$. Since $[X^L,X_1^L]=cX_2^L+X^L_r$ for a certain constant $c$ and an element $X^L_r\in V_r^L$, one obtains that $\mathcal{L}_{X^L}J= X^L_r\wedge X_2^L$, which projects onto zero in $G_r/H$. By Lemma \ref{lemma3}, the $J$ is projectable onto $G_r/H$. Moreover, $[J,J]_{SN}$ is such that $\pi_*[J,J]_{SN}=0$ and $\pi_*J$ is a Poisson bivector. Hence, one obtains that the projection of $X^R$ onto $G_r/H$ is a Lie--Hamilton system relative to $\pi_*\Lambda$ that is locally diffeomorphic to a Lie--Hamilton system on the plane with a VG Lie algebra locally diffeomorphic to $V$.

It is now very simple to prove that the above generalisation allows us to recover the Lie--Hamilton systems on the plane related to the VG Lie algebras P$_2$, P$_3$, I$_4$, and I$_5$.  

It seems to us that the previous procedure could be extended to study Lie--Hamilton systems on higher-dimensional manifolds. The idea is to construct from a Lie--Hamilton system $X$ on a manifold $N$ with a locally transitive VG Lie algebra $V$ of Hamiltonian vector fields an automorphic Lie system on a certain Lie group $G$ with a VG Lie algebra of right-invariant vector fields isomorphic to $V$. Then, a left-invariant bivector field $\Lambda$ must become projectable onto the quotient space with an isotropy group, which can be tested by using the algebraic properties of $\mathfrak{g}$ and/or using $\mathcal{T}^L(G)$ as previously. 

\section{Conclusions and outlook}
We have developed two natural geometric models for the study of Lie--Hamilton systems on the plane that, quite probably, can also be used to study higher-dimensional Lie--Hamilton systems. We plan to study such an extension in the future. 

On the other hand, it is natural to wonder whether the given approaches to Lie--Hamilton systems are useful to study their superposition rules. Recall that superposition rules for certain Lie--Hamilton systems were obtained  through Casimir functions (see \cite{LS19} and references therein), while Winternitz and coworkers derived superposition rules for classes of Lie systems  by considering them as projections of  automorphic Lie systems (see \cite{Dissertationes,LS19,SW84,SW84II}). Hence, it seems to us that our techniques, where such structures appear naturally, can be useful to study the existence and derivation of superposition rules. We aim to study this topic in further works.

\vspace{6pt} 

\section*{Acknowledgements}

J. de Lucas acknowledges support from the contract 623 of the Faculty of Physics of the University of Warsaw. 
This work is circumscribed within the research topics carried out in the project HARMONIA 2016/22/M/ST1/00542 financed by the Polish National Science Center (POLAND). We would like to thank the anonymous referees for the careful reading of the work, raising interesting questions that led to improve the paper, proposing interesting further research tasks, and pointing out the references \cite{Ma01,Um89}.

\end{document}